\documentclass{article}

\usepackage{hyperref}
\usepackage{graphicx}
\usepackage{amsmath}
\usepackage{amsthm}
\usepackage{enumerate}
\usepackage{fullpage}

\newtheorem{theorem}{Theorem}[section]
\newtheorem{lemma}[theorem]{Lemma}

\title{Multiple source, single sink maximum flow in a planar graph}
\author{Glencora Borradaile \\ Oregon State University \and Christian
  Wulff-Nilsen \\ University of Copenhagen}

\begin{document}

\maketitle
\abstract{We give an $O(n^{1.5}\log n)$ time algorithm for finding the
  maximum flow in a directed planar graph with multiple sources and a
  single sink.  The techniques generalize to a subquadratic
  time algorithm for bounded genus graphs.
}

\section{Introduction}
\label{sec:introduction}

In general graphs, multiple source flow problems are reduced to
single-source problems but connecting a super source to the sources
with infinite capacity arcs.  In planar graphs this reduction destroys
the planarity.  Using the maximum flow algorithms for general graphs,
a multiple-source, single (or multiple) sink max flow in directed
planar graphs can be solved in $O(n^2 \log n)$ time using Goldberg and
Tarjan's preflow push algorithm~\cite{GT88} or $O(n^{1.5}\log n \log
U)$ where $U$ is the maximum edge capacity using Goldberg and Rao's
binary blocking flow algorithm~\cite{GR98}.  In this paper we give an
$O(n^{1.5} \log n)$ algorithm for the problem in planar graphs by
combining preflow push and augmenting path algorithms.

In planar graphs, multiple source and sink flow problems have been
studied by Miller and Naor, giving subquadratic-time algorithms for
the case when the sources and sinks are on a common face and for the
feasibility problem (the amount of flow out of every source and into
every sink is known)~\cite{MN95}.  The maximum single source, single
sink flow in a directed planar graph can be found in $O(n \log n)$
time~\cite{BK09}.

\section{Definitions}
\label{sec:definitions}

We are given a directed planar graph $G$ with arc capacities a set of
source vertices $S$ and sink vertices $T$.  A flow is an assignment of
values to arcs not exceeding the capacity such that the flow entering
a non-source, non-sink vertex is equal to the flow leaving the same
vertex.  A flow is maximum if it maximizes the amount of flow leaving
$S$ and (equivalently) there is no residual path from any vertex in
$S$ to any vertex in $T$.  An arc is residual if the flow is less than
the capacity.  The reverse of an arc is residual if there is flow on
the arc.  For more formal definitions, see~\cite{BK09}.

A {\em preflow} is a flow that allows excess inflow at vertices that
are not sink vertices.  A maximum preflow is a preflow that maximizes
the flow into the sinks.  We will use the following lemma, which holds
for general graphs. The forward direction follows from the definition
of maximum preflows. The reverse direction follows from the Max Flow,
Min Cut Theorem: if there are no residual paths from sources or
vertices with excess inflow to sinks, then there is a saturated cut
separating the sources from the sinks and the preflow cannot be
increased (see~\cite{GT86}).

\begin{lemma}\label{lem:preflow-equiv}
  A preflow is maximum if and only if there is no residual path from a
  source to a sink or from a vertex with excess inflow to a sink.
\end{lemma}

We also assume for this draft that the reader is familiar with
recursive subdivisions of planar graphs. A {\em piece} is a subgraph
resulting from this decomposition. See~\cite{FR06} for details.

\section{Algorithm for a piece not containing sinks}
\label{sec:algorithm}

Our multiple-source, single sink algorithm
(Section~\ref{sec:overall-algorithm}) is a recursive algorithm using a
recursive decomposition based on cycle separators.  In this section,
we present a solution to the non-trivial subproblem in which we need
to find a maximum flow from a set of sources in a piece (of the
recursive decomposition) to sinks outside the piece.  More formally:
\begin{description} 
\item[Input:] a piece $P$ with a constant number of holes, a set $S_P$
  of sources in $P$, and a set $T_P$ of sinks none of which are
  internal vertices of $P$.
\item[Output:] The max flow from $S_P$ to $T_P$.
\end{description}

The algorithm proceeds in phases, mimicking a preflow algorithm in two
{\em pushes}: one to the boundary of the piece (Phase 1) and one from
the boundary to the sinks (Phase 2).  In a final phase, the (maximum)
preflow is converted to a maximum flow.

\begin{description}
\item[Phase 1:] Find the max flow in $P$ from $S_P\setminus\partial P$
  to $\partial P$.

\item[Phase 2:] For each vertex $p$ in $\partial P$, if $p$ is a
  source, compute a max $pT_P$-flow. Otherwise, push as much flow from
  $p$ to $T_P$ as possible while not exceeding the inflow to $p$.

\item[Phase 3:] Convert the above maximum preflow into a maximum flow.
\end{description}

\subsection{Correctness}
\label{sec:correctness}

We need the following lemma, which holds for general graphs.
\begin{lemma}\label{lem:cuts-stay-saturated}
  In a graph, let $A,B$ be a partition of the vertices such that all
  the arcs from $A$ to $B$ are saturated and the sinks $t$ are
  vertices of $B$.  Augmenting a preflow cannot make this cut
  residual.
\end{lemma}

\begin{proof}
  By assumption, there cannot be an augmenting path from a vertex in
  $A$ to a sink $t$, so consider an augmenting path $P$ from a vertex
  $b\in B$ to $t$. Assume for the sake of contradiction that $P$ makes
  the cut residual. Then it would have to intersect $A$. Since $b,t\in
  B$, there is an edge $(u,v)$ on $P$ such that $u\in A$ and $v\in B$.
  But by assumption, this edge is saturated, which is a contradiction.
  This shows the lemma.
\end{proof}

We show that after Phase 2, a maximum preflow has been computed. Let
$S_P' = S_P\setminus\partial P$ and $t\in T_P$.

After Phase 1, there are no residual $s$-to-$t$ paths for any $s \in
S_P'$: such a path would have, as a prefix, a residual $s$-to-$p$ path
for a vertex $p \in \partial P$; this would contradict the maximality
of the $S_P'$-to-$\partial P$ flow.  At the end of Phase 1, there is a
saturated $S_P'-\partial P$ cut $C_{S_P'}$ (separating $S_P'$ from
$\partial P$). By Lemma~\ref{lem:cuts-stay-saturated}, the
augmentations in Phase 2 cannot introduce a residual $s$-to-$t$ path
for any $s \in S_P'$.

It remains to show that at the end of Phase 2 there are no residual
paths to $t$ from a boundary vertex $p$ which is either a source or
which has excess inflow; by Lemma~\ref{lem:preflow-equiv}, the flow at
the end of Phase 2 is a maximum $S_PT_P$-preflow.  If $p$ is a source,
we saturate all edges in a $pt$-cut when finding a max $pT_P$-flow in
Phase 2. By Lemma~\ref{lem:cuts-stay-saturated}, this cut will stay
saturated.  Now, assume that $p$ has excess inflow $f_p$ at the start
of Phase 2.  There are two cases. If we are able to route $f_p$ units
of flow to $t$, then $p$ no longer has excess inflow.  If we are
unable to route $f_p$ units of flow to $t$, there is a saturated $pt$
cut $C_p$ that is a witness to there being no residual $p$-to-$t$
paths. Lemma ~\ref{lem:cuts-stay-saturated} implies that augmenting
flows from other boundary vertices cannot make this cut residual.

\subsection{Running time}
\label{sec:running-time}

We analyze the running time of Phases 2 and 3. Phase 1 will be
computed recursively; the analysis of the recursive algorithm is in
Section~\ref{sec:overall-algorithm}.


Phase 2 can be solved in $O(|\partial P|n\log n) = O(\sqrt{|P|}n\log
n)$ time by computing a maximum single-source, single-sink flow (in
the entire graph) for each $p \in \partial P$ and $t \in T_P$.  If $p$
is a source then compute the max $pt$-flow algorithm for each $t\in
T_P$. If $p$ is a non-source vertex with excess $f$ from Phase 1,
augment the graph with a source $s$ connected to $p$ by an arc with
capacity $f$ and compute the maximum $st$ flow for each $t \in T_P$.
Each max flow computation takes $O(n\log n)$ time in planar
graphs. Hence, Phase 2 runs in $O(\sqrt{|P|}n\log n)$ time.

Maximum preflows can be converted into maximum flows in sparse graphs
in $O(n \log n)$ time (see Goldberg and Tarjan~\cite{GT86,GT88}; the
conversion is only made explicit in the earlier conference version):
first eliminate cycles of flow ($O(n \log n)$ time) and then eliminate
excess inflow from vertices by processing them in reverse topological
order in the flow's support graph ($O(n)$ time).  In fact, our maximum
preflow is already acyclic, since the flows computed in Phase 2 are
acyclic; converting to a max flow will, in fact, only take linear
time.

The total time required for Phases 2 and 3 is therefore
$O(\sqrt{|P|}n\log n)$.

\section{Multiple-source, single-sink algorithm}
\label{sec:overall-algorithm}

In this section, the overall max flow algorithm is described.
Since we apply recursion, we need to consider a slightly more
general problem than multiple source, single sink max flow in
order to ensure that subproblems are of the same form:
we need to find a max flow from a set of sources to at most a
constant $t$ number of sinks.

Let us regard the entire graph as a piece $P$ with no boundary vertices.
We define each sink to be a boundary vertex of $P$. Each such sink now
defines a degenerate hole of $P$ and $P$ contains no sinks in its interior.
By assumption, the number of sinks is bounded by $t$ so the number of holes
of $P$ is at most $t$.

Next, we obtain a subdivision of $P$ into a constant
number $p$ of subpieces; each subpiece has at most $c_pn$ vertices and
$\sqrt{c_pn}$ boundary vertices (constant $c_p < 1$). This is done with
the recursive $r$-division algorithm of Frederickson (see Lemmas 1 and 2
in~\cite{Frederickson87}) but with Miller's cycle separator
theorem~\cite{Miller86} instead of the separator theorem of Lipton and
Tarjan~\cite{LT79}. Through the recursion, we
ensure that the number of holes in each piece is bounded by $t - 1$. If we
pick $t$ to be a sufficiently large constant, we can ensure this bound
with the approach introduced by Fakcharoenphol and
Rao~\cite{FR06}.

We now run the algorithm of the previous section on each subpiece $P'$.
To solve Phase 1 for $P'$, we define a new planar graph $P''$ from $P'$
by introducing a super sink for each hole and adding an edge of infinite
capacity from each boundary vertex of that hole to the super sink. The
same is done for the external face of $P'$ if it contains boundary
vertices. Now, solving Phase 1 for $P'$ corresponds to finding a max flow
from the sources in $P''$ to its super sinks. Since $P'$ has at most
$t - 1$ holes, $P''$ has at most $t$ super sinks and so we can recurse
on $P''$ to solve this problem (note that only the super sinks of $P''$
and not the sinks of $P$ belonging to $P''$ are regarded as sinks in
the recursive call).

\subsection{Running Time}
\label{sec:running-time-1}

As described in Section~\ref{sec:running-time}, Phases 2 and 3 run in
$O(n^{3/2}\log n)$ time.

We define a recurrence relationship to bound the time for Phase 1 and
the entire algorithm. For $n$ larger than some constant, the size of
each piece in the division (including the super sinks added to the
piece) is at most $c_p'n$, $c_p < c_p' < 1$.

We repeat the recursion until each piece has size at most some
constant $r$. Let $N$ be the total size of all pieces at the
leaves of the recursion tree. Using ideas of
Frederickson~\cite{Frederickson87} (see Lemmas 1 and 2 and their
proofs in that paper for details), $N = kn$ for some constant $k$.

In the analysis, we assume that the total size of pieces at any recursion
level is $N$ by regarding the vertices in pieces at the leaves of the
recursion tree to be present in every level. This allows us to regard
the pieces in any level as being pairwise vertex-disjoint by counting a
vertex according to its multiplicity, once for each piece to which it
belongs. Furthermore, we may assume that no new vertices
(i.e.~super sinks) are introduced in the recursive steps since all
$N$ vertices (including all added super sinks) are present in any
recursion level. Finally, we may assume that the size of a piece
in a division is exactly $c_p'n$ since larger pieces will only
increase running time.

From these assumptions, it follows that we may restrict our attention to the
case $p = 1/c_p'$ and the recurrence relation for the running time of our
algorithm becomes \[
T(N) \leq pT(N/p) + c'N^{3/2}\log N, \] where $c'$ is a constant and
the second term is the total time for Phases 2 and 3 and for computing
the division of the graph. We will show that $T(N)\leq cN^{3/2}\log N$
for some constant $c$ (which we may assume is true for small $N$). For the
induction step, we have \[
  T(N) \leq pc(N/p)^{3/2}\log N + c'N^{3/2}\log N
           = (c/\sqrt p + c')N^{3/2}\log N.
\]
The desired time bound is achieved by setting $c = c'/(1 - 1/\sqrt
p) > 0$.

\subsection{Correctness}
\label{sec:correctness-1}

Correctness of the overall algorithm follows from
Theorem~\ref{thm:multi-FF}: that is, the algorithm is correct over all
pieces because we can consider the sources in any order, and we always
fully saturate the flow from the source we consider.

\begin{theorem}\label{thm:multi-FF}
  Given an instance of a multiple source, multiple sink maximum flow
  problem, one can find the maximum flow by:
  \begin{tabbing}
    iterating over the sources $s$ in any order\\
    \qquad iterating over the sinks $t$ in any order \\
    \qquad \qquad saturating a maximum $st$-flow in the residual graph
  \end{tabbing}
\end{theorem}

\begin{proof}
  We prove by induction over the inner and outer loops of the
  algorithm.  Let $s_1,s_2,\ldots$ be the arbitrary order that the
  sources are iterated over.

  {\em Outer loops:} Say we are currently considering source $s_j$.
  Assume that there are no residual paths from source $s_i$ for any $i
  < j$ to any sink.  We will show that no residual $s_j$-to-$t$ path
  can, by way of augmentation, introduce an $s_i$-to-$t'$ residual
  path.  By the inductive hypothesis, the Max Flow, Min Cut Theorem,
  and submodularity of cuts, there are non-residual non-crossing cuts
  $S_i, \bar S_i$ and $S_i', \bar S_i'$ with $s_i \in S_i \cap S_i'$
  and $t \in \bar S_i \cap \bar S_i'$.  Let $R$ be a residual
  $s_j$-to-$t$ path.  If $R$ could introduce a residual $s_i$-to-$t'$
  path, then it must be that $s_j \in \bar S_i'$ and $t\in S_i$.  By
  the submodularity of cuts, it must be that $S_i \subset S_i'$.  It
  follows that $S_i,\bar S_i$ is a non-residual cut that also
  separates $s_i$ and $t'$ and augmenting $R$ cannot introduce a
  residual $s_i$-to-$t'$ path.

  {\em Inner loops:} Suppose we are currently considering source
  $s_\ell$.  Let $t_1,t_2, \ldots$ be the arbitrary order that the
  sinks are iterated over for source $s_\ell$.  Suppose we are
  saturating the max $s_\ell t_j$ flow.  By the above argument, doing
  so will not introduce a residual source-to-sink path for any source
  $s_k$, $k < \ell$.  Assume that there are no residual
  $s_\ell$-to-$t_i$ paths for any $i < j$.  By this inductive
  hypothesis and the Max Flow, Min Cut Theorem, there are non-residual
  cuts $S_i, \bar S_i$ such that $s_\ell \in S_i$ and $t_i \in \bar
  S_i$.  If there is a non-zero $s_\ell$-to-$t_j$ flow in the residual
  graph, then it must be that $t_j \in S_i$ for every $i < j$.  It
  follows that no $s_\ell$-to-$t_j$ flow can leave $S_i$.    
\end{proof}

Note that saturating an arbitrary order of the source, sink pairs will
in general not result in a maximum flow, see
Figure~\ref{fig:any-pair-counterexample}.

\begin{figure}[ht]
  \centering
  \includegraphics{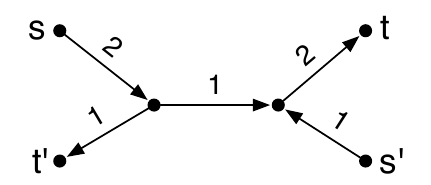}
  \caption{A counterexample to arbitrarily saturating source sink
    pairs: the maximum $\{s,s'\}$-to-$\{t,t'\}$ flow has value 3.  If
    the source, sink pairs are maximized in the order
    $(s,t),(s',t'),(s,t'),(s',t)$, the value of the flow found is only
    2.}
  \label{fig:any-pair-counterexample}
\end{figure}

\section{Bounded Genus Graphs}
\label{sec:bounded-genus}

Since our algorithm does not rely too much on planarity, it is easy to
generalize it to bounded genus graphs. For such graphs, we can use the
$O(n\log^2 n\log^2 C)$ time max $st$-flow algorithm due to Erickson,
Chambers and Nayyeri~\cite{ECN09} ($C$ is the sum of capacities).  To
begin the decomposition, we may first planarize the graph with a set
of $g$ simple cycles, each of length $O(\sqrt{n})$ (e.g.~construction
by Hutchinson and Miller~\cite{HM86}).  This will allow us to continue
with a decomposition based on planar separators and introduce only
$O(g)$ holes in which we can embed a super sink while maintaining
planarity as required for Phase 1.  This gives an $O(n^{3/2}\log^2 n
\log^2 C)$ time algorithm; of course, using the binary blocking flow
algorithm would be asymptotically faster.

\bibliographystyle{plain}
\bibliography{long,multi-bib}

\begin{thebibliography}{10}

\bibitem{BK09}
G.~Borradaile and P.~Klein.
\newblock An ${O}(n \log n)$ algorithm for maximum st-flow in a directed planar
  graph.
\newblock {\em Journal of the ACM}, 56(2):1--30, 2009.

\bibitem{ECN09}
J.~Erickson, E.~Chambers, and A.~Nayyeri.
\newblock Homology flows, cohomology cuts.
\newblock In {\em Proceedings of the 41st Annual ACM Symposium on Theory of
  Computing}, pages 273--282, 2009.

\bibitem{FR06}
J.~Fakcharoenphol and S.~Rao.
\newblock Planar graphs, negative weight edges, shortest paths, and near linear
  time.
\newblock {\em J. Comput. Syst. Sci.}, 72(5):868--889, 2006.

\bibitem{Frederickson87}
G.~Frederickson.
\newblock Fast algorithms for shortest paths in planar graphs with
  applications.
\newblock {\em SIAM Journal on Computing}, 16:1004--1022, 1987.

\bibitem{GR98}
A.~Goldberg and S.~Rao.
\newblock Beyond the flow decomposition barrier.
\newblock {\em Journal of the ACM}, 45(5):783--797, 1998.

\bibitem{GT86}
A.~Goldberg and R.~Tarjan.
\newblock A new approach to the maximum-flow problem.
\newblock In {\em Proceedings of the 18th Annual ACM Symposium on Theory of
  Computing}, pages 136--146, 1986.

\bibitem{GT88}
A.~Goldberg and R.~Tarjan.
\newblock A new approach to the maximum-flow problem.
\newblock {\em Journal of the ACM}, 35(4):921--940, 1988.

\bibitem{HM86}
J.~Hutchinson and G.~Miller.
\newblock Deleting vertices to make graphs of positive genus planar.
\newblock {\em Discrete Algorithms and Complexity Theory}, 1:81--98, 1986.

\bibitem{LT79}
R.~J. Lipton and R.~E. Tarjan.
\newblock A separator theorem for planar graphs.
\newblock {\em SIAM Journal on Applied Mathematics}, 36(2):177--189, 1979.

\bibitem{Miller86}
G.~L. Miller.
\newblock Finding small simple cycle separators for 2-connected planar graphs.
\newblock {\em Journal of Computer and System Sciences}, 32(3):265--279, 1986.

\bibitem{MN95}
G.~L. Miller and J.~Naor.
\newblock Flow in planar graphs with multiple sources and sinks.
\newblock {\em SIAM Journal on Computing}, 24(5):1002--1017, 1995.

\end{thebibliography}

\end{document}